\documentclass[11pt]{article}

\usepackage{amsmath,amssymb,amsbsy,amsfonts,amsthm,latexsym,
            amsopn,amstext,amsxtra,euscript,amscd,color}
            
\PassOptionsToPackage{hyphens}{url}\usepackage{hyperref}

\newtheorem{remark}{Remark}

\newtheorem{thm}{Theorem}
\newtheorem{lem}[thm]{Lemma}

\newtheorem{defn}[thm]{Definition}

\newcommand{\Tr}{{\rm Tr}}
\newcommand{\Trn}{{\rm Tr}_n}

\def\+{\oplus}

\def\F{{\mathbb F}}

\def\F{{\mathbb F}}

\def\00{{\bf 0}}
\def\11{{\bf 1}}
\def\+{\oplus}

\def\\{\cr}
\def\({\left(}
\def\){\right)}

\providecommand{\newoperator}[3]{%
  \newcommand*{#1}{\mathop{#2}#3}}

\newoperator{\FD}{\mathrm{FD}}{\nolimits}

\usepackage[colorinlistoftodos,prependcaption,textsize=tiny]{todonotes}

\begin{document}
\title{\bf Investigations on $c$-(almost) perfect nonlinear functions}
\author{Constanza Riera\thanks{C. Riera is with Department of Computer Science,
 Electrical Engineering and Mathematical Sciences,   Western Norway University of Applied Sciences,  5020 Bergen, Norway. 
E-mail: csr@hvl.no},  Pantelimon~St\u anic\u a
\thanks{P. St\u anic\u a is with Applied Mathematics Department, Naval Postgraduate School, Monterey 93943, USA. 
E-mail: pstanica@nps.edu}}

\maketitle

\begin{abstract}
In a prior paper~\cite{EFRST20}, along with P. Ellingsen, P. Felke and A. Tkachenko,  we defined a new (output) multiplicative differential, and the corresponding $c$-differential uniformity, which has the potential of extending differential cryptanalysis. Here, we continue the work, by looking at some APN functions through the mentioned concept and show that their $c$-differential uniformity increases significantly, in some cases.     
\end{abstract}


\section{Introduction and motivation}

In~\cite{BCJW02}, the authors used a new type of differential that is quite useful from a practical perspective for ciphers that utilize modular multiplication as a primitive operation. It is an extension of a type of differential cryptanalysis and it was used to cryptanalyze some existing ciphers (like a variant of the well-known IDEA cipher).
The authors argue that one should look (and some authors did) at other types of differentials for a Boolean (vectorial) function $F$, not only the usual $\left(F(x+a),F(x)\right)$. In~\cite{BCJW02}, the differential used in their attack was $\left(F(cx),F(x) \right)$.
Drawing inspiration from the mentioned  successful attempt, along with P. Ellingsen, P. Felke and A. Tkachenko (see ~\cite{EFRST20}),  we defined a new (output) multiplicative differential, and the corresponding generalized differential uniformity.  
In this paper, we investigate some known APN functions under this new multiplicative differential, and show that their $c$-differential uniformity drops significantly, in some cases.     

The objects of this study are Boolean and $p$-ary functions (where $p$ is an odd prime) and some of their differential properties.  We will introduce here only some needed notation, and the reader can consult~\cite{Bud14,CH1,CH2,CS17,MesnagerBook,Tok15} for more on Boolean and $p$-ary functions.

Let $n$ be a positive integer and $\F_{p^n}$ denote the  finite field with $p^n$ elements, and $\F_{p^n}^*=\F_{p^n}\setminus\{0\}$ is the multiplicative group (for $a\neq 0$, we often write $\frac{1}{a}$ to mean the inverse of $a$ in the multiplicative group). We let $\F_p^n$ be the $n$-dimensional vector space over $\F_p$. We will denote by $\eta(\alpha)$ the quadratic character of $\alpha$ (that is, it is $\eta(\alpha)=0$ if $\alpha=0$, $\eta(\alpha)=1$ if $0\neq \alpha$ is a square, $\eta(\alpha)=-1$ if $0\neq \alpha$ is not a square). $|A|$ will denote the cardinality of a set $A$.
We call a function from $\F_{p^n}$ to $\F_p$  a {\em $p$-ary  function} on $n$ variables. 
$\Trn:\F_{p^n}\to \F_p$ is the absolute trace function, given by $\Trn(x)=\sum_{i=0}^{n-1} x^{p^i}$.
 
Given a $p$-ary  function $f$, the derivative of $f$ with respect to~$a \in \F_{p^n}$ is the $p$-ary  function
$
 D_{a}f(x) =  f(x + a)- f(x), \mbox{ for  all }  x \in \F_{p^n}.
$

For positive integers $n$ and $m$, any map $F:\F_{p^n}\rightarrow\F_{p^m}$ (or, alternatively, $F:\F_p^n\to\F_p^m$, though we will in this paper deal with the former form of the function) is called a {\em vectorial $p$-ary  function}, or {\em $(n,m)$-function}. When $m=n$, $F$ can be uniquely represented as a univariate polynomial over $\F_{p^n}$ (using some identification, via a basis, of the finite field with the vector space) of the form
$
F(x)=\sum_{i=0}^{p^n-1} a_i x^i,\ a_i\in\F_{p^n},
$
whose {\em algebraic degree}   is then the largest Hamming weight of the exponents $i$ with $a_i\neq 0$.

For an $(n,n)$-function $F$, and $a,b\in\F_{p^n}$, we let $$\Delta_F(a,b)=|{\{x\in\F_{p^n} : F(x+a)-F(x)=b\}}|.$$ We call the quantity
$\Delta_F=\max\{\Delta_F(a,b)\,:\, a,b\in \F_{p^n}, a\neq 0 \}$ the {\em differential uniformity} of $F$. If $\Delta_F= \delta$, then we say that $F$ is differentially $\delta$-uniform. If $\delta=1$, then $F$ is called a {\em perfect nonlinear} ({\em PN}) function, or {\em planar} function. If $\delta=2$, then $F$ is called an {\em almost perfect nonlinear} ({\em APN}) function. It is well known that PN functions do not exist if $p=2$.

Inspired by a practical differential attack developed  in~\cite{BCJW02} (though, via a different differential), we extended the definition of derivative and differential uniformity in ~\cite{EFRST20}, in the following way:
Given a $p$-ary $(n,m)$-function   $F:\F_{p^n}\to \F_{p^m}$, and $c\in\F_{p^m}$, the ({\em multiplicative}) {\em $c$-derivative} of $F$ with respect to~$a \in \F_{p^n}$ is the  function
\[
 _cD_{a}F(x) =  F(x + a)- cF(x), \mbox{ for  all }  x \in \F_{p^n}.
\]
(Note that, if   $c=1$, then we obtain the usual derivative, and, if $c=0$ or $a=0$, then we obtain a shift of the function.)

\begin{sloppypar}
For an $(n,n)$-function $F$, and $a,b\in\F_{p^n}$, we let $_c\Delta_F(a,b)=|{\{x\in\F_{p^n} : F(x+a)-cF(x)=b\}}|$. In the following, we call the quantity
$_c\Delta_F=\max\left\{_c\Delta_F(a,b)\,:\, a,b\in \F_{p^n}, \text{ and } a\neq 0 \text{ if $c=1$} \right\}$ 
the {\em $c$-differential uniformity} of~$F$. If $_c\Delta_F=\delta$, then we say that $F$ is differentially $(c,\delta)$-uniform. If $\delta=1$, then $F$ is called a {\em perfect $c$-nonlinear} ({\em PcN}) function (certainly, for $c=1$, they only exist for odd characteristic $p$; however, as proven in ~\cite{EFRST20}, there exist PcN functions for $p=2$, for all  $c\neq1$). If $\delta=2$, then $F$ is called an {\em almost perfect $c$-nonlinear} ({\em APcN}) function. 
When we specify the constant $c$ for which the function is PcN or APcN, then we may use the notation $c$-PN, or $c$-APN.
It is easy to see that if $F$ is an $(n,n)$-function, that is, $F:\F_{p^n}\to\F_{p^n}$, then $F$ is PcN if and only if $_cD_a F$ is a permutation polynomial.
\end{sloppypar}

NB: Recently, in an independent work, Bartoli and Timpanella~\cite{BT}   gave a generalization of planar functions as follows.
\begin{defn}\label{defBT}
Let $\beta \in \mathbb{F}_{p^n} \backslash \{0,1\}$. A function $F : \mathbb{F}_{p^n} \rightarrow \mathbb{F}_{p^n}$ is a $\beta$-planar  function in $\mathbb{F}_{p^n}$ if $\forall~ \gamma \in \mathbb{F}_{p^n},~~~~ F(x+\gamma) - \beta F(x)$ is a permutation of $\mathbb{F}_{p^n}.$
\end{defn}
In the particular case, when $\beta =-1$, then $\beta$-planar function is called quasi-planar. In view of the definitions of ~\cite{EFRST20}, the $\beta$-planar functions are simply  PcN functions and quasi-planar functions are PcN functions with $c=-1$.

In this paper we continue the investigation on the $c$-differential uniformity of the Gold function (resolving some computational observations from~\cite{EFRST20}) and some of the functions from the Helleseth-Rong-Sandberg table and related ones. 
For example, we showed that $x^{\frac{p^k+1}{2}}$ is PcN with respect to $c=-1$ if and only if $\frac{2n}{\gcd(2n,k)}$ is odd, as well $x^{\frac{3^n+3}{2}}$ is PcN (respectively APcN) with respect to $c=-1$, if $n$ is odd (respectively, even).  

\noindent {\em Note added after posting on arxiv}: We were made aware that the first part of  our Theorems~\ref{thm:Gold} and~\ref{plusone} also appeared concurrently  in~\cite{YZ20}.

\section{Prior results on $c$-differential uniformity}
 
One of our major theorems from~\cite{EFRST20} dealt with the known PN functions. 
\begin{thm}[\textup{\cite{EFRST20}}]
\label{c-diff}
Let $F:\F_{p^n}\to \F_{p^n}$ be the monomial $F(x)=x^d$,  and $c\neq  1$ be fixed. The following statements hold:
\begin{enumerate}
\item[$(i)$] If $d=2$, then $F$ is  APcN, for all $c\neq 1$.
\item[$(ii)$] If $d=p^k+1$, $p>2$, then  $F$ is not PcN, for all $c\neq 1$. Moreover, when $(1-c)^{p^k-1}=1$
and ${n}/{\gcd{(n,k)}}$ is even, the $c$-differential uniformity   $_c\Delta_F\geq p^g+1$, where $g=\gcd(n,k)$.
\item[$(iii)$]   Let $p=3$. If $\displaystyle d=\frac{3^k+1}{2}$, then  $F$ is   PcN, for $c=-1$ if and only if $\displaystyle \frac{n}{\gcd(n,k)}$ is odd.
\item[$(iv)$] If $p=3$ and $F(x)=x^{10}-u x^6-u^2 x^2$,  the $c$-differential uniformity of $F$ is $_c\Delta_F\geq 2$, for $c\neq 1$.
\end{enumerate}
\end{thm}

  The $c$-differential uniformity of the inverse function has also been investigated~\cite{EFRST20} and a thorough description was obtained for all values of $c$, both in the even and odd case. 
   
   Since it will be used throughout, we state here~\cite[Lemma 9]{EFRST20}.   
  \begin{lem}
\label{lem:gcd}
Let $p,k,n$ be integers greater than or equal to $1$ (we take $k\leq n$, though the result can be shown in general). Then
\begin{align*}
&  \gcd(2^{k}+1,2^n-1)=\frac{2^{\gcd(2k,n)}-1}{2^{\gcd(k,n)}-1},  \text{ and if  $p>2$, then}, \\
& \gcd(p^{k}+1,p^n-1)=2,   \text{ if $\frac{n}{\gcd(n,k)}$  is odd},\\
& \gcd(p^{k}+1,p^n-1)=p^{\gcd(k,n)}+1,\text{ if $\frac{n}{\gcd(n,k)}$ is even}.\end{align*}
Consequently, if either $n$ is odd, or $n\equiv 2\pmod 4$ and $k$ is even,   then $\gcd(2^k+1,2^n-1)=1$ and $\gcd(p^k+1,p^n-1)=2$, if $p>2$.
\end{lem}

 \section{The $c$-differential uniformity  of the Gold function}

 In~\cite{EFRST20} we found the $c$-differential uniformity of the Gold function $x\mapsto x^{p^k+1}$ for odd characteristic~$p$. Regarding even characteristic, 
 we observed in the same paper that, if $3\leq n\leq 8$, the $c$-differential uniformity of the Gold $x\mapsto x^5$ and Kasami $x\mapsto x^{13}$  functions is $3$ for $n$ odd and $5$ for $n$ even. 
It was proposed there that it would be interesting to investigate the situation for all values of $n$, especially, since if $c=1$, the result is well-known. It is our goal here to answer the question and surprisingly, reveal that the $c$-differential uniformity of these functions may increase significantly. In this paper we deal with the Gold function and its $c$-differential uniformity.
  \begin{thm}
  \label{thm:Gold}
 Let $2\leq k<n$, $n\geq 3$ and $G(x)=x^{2^k+1}$ be the Gold function on $\F_{2^n}$ and $1\neq c\in\F_{2^n}$. Assume that $n=md$, where $d=\gcd(n,k)$, and $m\geq 3$, when $n$ is odd, respectively, $m\geq 4$, when $n$ is even. Then, the $c$-differential uniformity of $G$ is $_c\Delta_G=2^d+1$.
  \end{thm}
  \begin{proof}
  We consider the differential equation at $a$, say $G(x+a)-c\,G(x)=b$, for some $b\in\F_{2^n}$, which is equivalent to
  \[
  (1-c)\, x^{2^k+1}+ x^{2^k} a+x\, a^{2^k} +a^{2^k+1}-b=0.
  \]
  Dividing by $(1-c)$ and taking $x=y-\frac{a}{1-c}$ this last equation transforms into
  \[
  y^{2^k+1}+ \frac{a^{2^k}}{1-c} \left(1+\frac{1}{(1-c)^{2^k-1}} \right)y+\frac{c a^{2^k+1}+b (1-c)}{(1-c)^2}=0.
  \]
  Now, let $y=\alpha z$, where $\displaystyle \alpha=\left( \frac{a^{2^k}}{1-c}\left(1+\frac{1}{(1-c)^{2^k-1}} \right)\right)^{2^{-k}}$ (the $2^k$-root exists since $\gcd(2^k,2^n-1)=1$). The previous equation becomes
  \begin{equation}
  \label{eq:cGold}
  z^{2^k+1}+z+\beta=0,
  \end{equation}
  where $\displaystyle  \beta=\frac{c a^{2^k+1}+b (1-c)}{\alpha^{2^k+1}(1-c)^2}$. 
  
  We will be using some results of~\cite{HK08} (see also~\cite{Bluher04} and \cite{Dob06}).  
We first assume that $\gcd(n,k)=1$. By~\cite[Theorem 1]{HK08}, we know that Equation~\eqref{eq:cGold} has either none, one or three solutions in $\F_{2^n}$. In fact, the distribution of these cases for $n$ odd (respectively, $n$ even) is  (denoting by $M_m$ the amount of equations of type ~\eqref{eq:cGold} with $m$ solutions)
  \begin{align*}
  M_0&=\frac{2^n+1}{3}\ \left(\text{respectively}, \frac{2^n-1}{3}\right)\\
  M_1&=2^{n-1}-1\ \left(\text{respectively}, 2^{n-1}\right)\\
  M_3&=\frac{2^{n-1}-1}{3}\ \left(\text{respectively}, \frac{2^{n-1}-2}{3}\right).
  \end{align*}
 Then, for $n\geq3$, $c\neq1$, and $\gcd(n,k)=1$, and since $\beta$ is linear on $b$, this implies that, for any $\beta$ and any $a,c$, we can find $b$ such that $\displaystyle  \beta=\frac{c a^{2^k+1}+b (1-c) }{\alpha^{2^k+1}(1-c)^2}$, so the $c$-differential uniformity of the Gold function is $3$. 
  
  We now assume that $\gcd(n,k)=d>1$. As in~\cite{HK08}, for $v\in\F_{2^n}\setminus \F_{2^d}$, we denote $v_i:=v^{2^{ik}}$, $i\geq 0$, and we let  (for $n=md$)
  \[
  \begin{array}{l}
  C_1(x)=1\\
  C_2(x)=1\\
  C_{i+2}(x)=C_{i+1}(x)+x_iC_{i}(x)\ \mbox{ for } 1\leq i\leq n-1
  \end{array}
  \]
 Let $V=\frac{v_0^{2^{2k}+1}}{(v_0+v_1)^{2^k+1}}$. Then, by ~\cite[Lemma 1]{HK08}
  \[
   C_m(V)=\frac{\Tr^n_d(v_0)}{v_1+v_2}\prod_{j=2}^{m-1} \left( \frac{v_0}{v_0+v_1}\right)^{2^{jk}}.
  \]
   We know by~\cite[Lemma 1]{HK08} that if $n$ is odd (respectively, even) there are $\displaystyle \frac{2^{(m-1)d}-1}{2^{2d}-1}$ (respectively,  $\displaystyle \frac{2^{(m-1)d}-2^d}{2^{2d}-1}$) distinct zeros of $C_m(x)$ in $\F_{2^n}$, that are defined by $V$ above with $\Tr_d^n(v_0)=0$. Further, by ~\cite[Proposition 4]{HK08}, Equation~\eqref{eq:cGold} has $2^d+1$ zeros in $\F_{2^n}$ for as many as $M_{2^d+1}=\frac{2^{(m-1)d}-1}{2^{2d}-1}$, for $n$ odd, respectively, $M_{2^d+1}=\frac{2^{(m-1)d}-2^d}{2^{2d}-1}$, for $n$ even, values of $\beta$. To be more precise, those $\beta$ achieving this bound must satisfy $C_m(\beta)=0$ (if $d=k$, this is a complete description).  For $n$ odd,  $M_{2^d+1}\geq 1$ is achieved when $m\geq3$. For $n$ even,  $M_{2^d+1}\geq 1$ is achieved when $m\geq4$, and it is not true when $m=3$. 
  
  The only thing to argue now is whether for a fixed $c\neq 1$, and given $\beta\in\F_{2^n}$, there exist $a,b\in\F_{2^n}$ such that $\displaystyle  \beta=\frac{c a^{2^k+1}+b(1-c)}{\alpha^{2^k+1}(1-c)^2}$, where $\displaystyle \alpha^{2^{k}}=  \frac{a^{2^k}}{1-c}\left(1+\frac{1}{(1-c)^{2^k-1}} \right)$. However, that is easy to see since the obtained equation is linear in $b$. 
Therefore, the $c$-differential uniformity of $G$ is $2^d+1$.
  \end{proof}

  \section{Going through some entries in the Helleseth-Rong-Sandberg table and more}
  
  We display below some of the known examples of APN power functions in odd characteristic~\cite{HRS99,HS97}. 
  \begin{thm}
  \label{thm:HRS}
  Let $F(x)=x^d$ be a function over  $\F_{p^n}$, where $p$ is an odd prime. Then $F$ is an APN function if:
  \begin{itemize}
  \item[$(1)$] $d=3$, $p>3$;
  \item[$(2)$] $d=p^n-2$, $p>2$ and $p\equiv 2\pmod 3$.
  \item[$(3)$] $d=\frac{p^n-1}{2}-1$, $p\equiv 3,7\pmod {20}$, $p^n>7$, $p^n\neq 27$ and $n$ is odd;
  \item[$(4)$] $d=\frac{p^n+1}{4}+\frac{p^n-1}{2}$, $p^n\equiv 3\pmod 8$;
  \item[$(5)$] $d=\frac{p^n+1}{4}$, $p^n\equiv 7\pmod 8$;
  \item[$(6)$] $d=\frac{2p^n-1}{4}$, $p^n\equiv 2\pmod 3$;
\item[$(7)$]   $d=p^n-3$, $p=3$, $n>1$ odd;
\item[$(8)$] $d=p^m+2,p^m\equiv 1\pmod 3$, $n=2m$;
\item[$(9)$] $d=\frac{5^k+1}{2}, p=5$ and $\gcd(2n,k)=1$.
  \end{itemize}
  \end{thm}
  Dobbertin et al.~\cite{Dob03} pushed further the Helleseth-Rong-Sandberg table by explaining some entries (that is, shoving that some of the values $d$ that gave rise to APN functions were members of an infinite family of power functions), and showed that the differential uniformity of $F(x)= x^d$ over $\F_{3^n}$ is $\Delta_F\leq 2$ if 
 $$d=\begin{cases}
\frac{3^{(n+1)/2}-1}{2} & \text{if } n\equiv 3\pmod 4\\
\frac{3^{(n+1)/2}-1}{2} + \frac{3^n-1}{2} & \text{if } n\equiv 1\pmod 4,
\end{cases}
$$
(if $n=1,3$, $F$  is PN), as well as
$$d=\begin{cases}
\frac{3^{n+1}-1}{8} & \text{if } n\equiv 3\pmod 4\\
\frac{3^{n+1}-1}{8} + \frac{3^n-1}{2} & \text{if } n\equiv 1\pmod 4.
\end{cases}
$$
Later, in \cite{Leducq12}, Leducq proved that, in fact, for the functions above, $\Delta_F= 2$, that is, the functions above are APN.  The same paper proves that, over $\F_{5^n}$, for $l\leq2$ and $n\equiv -1\pmod 2^l$, the function $F$ is APN if $d=\frac{1}{2}\frac{5^{n+1}-1}{5^\frac{n+1}{2^l}+1}+\frac{5^n-1}{4}$.

 Dobbertin et al.~\cite{Dob03} conjectured also that over $\F_{5^n}$, $F$ is APN when $d=\frac{5^n-1}{4}+\frac{5^{n+1}/2-1}{2}$, if $n$ is odd, and this was subsequently shown by Zha and Wang~\cite{ZW11}.
 
 In this section, we will use Dickson polynomials of the first kind, which are defined as $\displaystyle D_d(x,a)= \sum_{i=0}^{\lfloor \frac{d}{2} \rfloor} \frac{d}{d-i}  \binom{d-i}{ i}(-a)^i x^{d-2i}$, and have the property $D_m\left(u+\frac{a}{u},a\right)=u^m+\left(\frac{a}{u}\right)^m$ for $u\in\F_{p^{2n}}$ \cite{LN96}; since, in this paper, the second variable is always 1, in the following, abusing notation, we will write $D_m(x)$). We will also use Theorem~9 of~\cite{CGM88}, which states that, for $p$ odd, and supposing $2^r||(p^{2n}-1)$ (where $2^r||t$ means that  $2^r|t$ but $2^{r+1}\nmid t$), then, for $x_0\in\F_{p^n}$,
$$
|D_d^{-1}(D_d(x_0))|=
\begin{cases}
m, &\mbox{ if }\eta(x_0^2-4)=1,\,D_d(x_0)\neq\pm2\\
\bar{\ell}, &\mbox{ if }\eta(x_0^2-4)=-1,\,D_d(x_0)\neq\pm2\\
\frac{m}{2}, &\mbox{ if }\eta(x_0^2-4)=1,\,2^t||d,\,1\leq t\leq r-2,D_d(x_0)=-2\\
\frac{\bar{\ell}}{2}, &\mbox{ if }\eta(x_0^2-4)=-1,\,2^t||d,\,1\leq t\leq r-2,D_d(x_0)=-2\\
\frac{m+\bar{\ell}}{2}, &\mbox{ otherwise,}
\end{cases}
$$
where $m=\gcd(d,p^n-1),\,\bar{\ell}=\gcd(d,p^n+1)$, and $\eta(\alpha)$ is the quadratic character of $\alpha$.

  In this section, we show that the $c$-differential uniformity of some of these functions will change for some (if not all) $c\neq1$. We start with item $(7)$ of Theorem~\ref{thm:HRS} (items $(1)$ and $(2)$ were dealt with in our paper~\cite{EFRST20}).
  \begin{thm}
  \label{thm:pn-3}
  Let $p=3$, $n\geq 2$ and $F(x)=x^{3^n-3}$ on $\F_{3^n}$.  If $c=-1$, the $c$-differential uniformity of  $F$ is $6$ for  $n\equiv 0\pmod 4$ and $4$, otherwise.  
  If $c=0$,  the $c$-differential uniformity of $F$ over $\F_{3^n}$ is $2$.
   If $c\neq 0,\pm1$, the $c$-differential uniformity of $F$  is  
  $\leq 5$. Moreover, the   $c$-differential uniformity of $4$ is attained for some $c$, for all $n\geq 3$, and the $c$-differential uniformity of $5$ is attained for all positive $n\equiv 0\pmod 4$.
  \end{thm}
  \begin{proof}
  For $a,b\in\F_{3^n}$, we look at the equation $F(x+a)-cF(x)=b$, that is,
  \begin{equation}
  \label{eq:thm6}
  (x+a)^{3^n-3}- c\, x^{3^n-3}=b.
  \end{equation}
  If $c=0$, the equation is then $(x+a)^{3^n-3}=b$. If $a=b=0$, then we get the unique solution $x=0$. If $a=0,b\neq 0$, the equation is then $bx^2=1$, which has two solutions if $b$ is a square and none, otherwise. If $a\neq 0,b=0$, the solution is $x=-a$. If $ab\neq 0$, the equation is then $b(x+a)^2=a$, which has two solutions if $a/b$ is a nonzero perfect square (always realizable). Summarizing, the function is APcN with respect to $c=0$. For the remainder of the proof, we assume that $c\neq 0,1$.
  
  If $a=b=0$, then $x=0$ is the only solution for Equation~\eqref{eq:thm6}.
  If $a=0$ and $b\neq 0$, then Equation~\eqref{eq:thm6} becomes $(1-c) x^{3^n-3}=b$. Surely, $x\neq 0$ and so, the equation becomes $bx^2=(1-c)$. Taking $\alpha\neq 0$ a fixed element of $\F_{3^n}$ and $b=\frac{1-c}{\alpha^2}$, then $x=\pm \alpha$ are solutions for this equation (observe that if $b\neq 0$ and $\frac{1-c}{b}$ is not a square in $\F_{3^n}$ there are no solutions for $(1-c) x^{3^n-3}=b$; for $c$ fixed, there are  $\frac{3^n-1}{2}$ nonzero values of $b$ such that $\frac{1-c}{b}$ is a square in $\F_{3^n}$). 
  
  If $a\neq 0$ and $b=0$, then $0\neq x\neq -a$, and Equation~\eqref{eq:thm6} becomes $(x+a)^{-2}= cx^{-2}$, that is, $\left(\frac{x}{x+a}\right)^2=c$, which has two solutions depending on whether $c$ is a nonzero perfect square or not (there are $\frac{3^n+1}{2}$ such $c$'s).  
  
  When $ab\neq 0$, we observe that $x=0$, $x=-a$ are solutions of Equation~\eqref{eq:thm6} if and only if $b=a^{-2}$, respectively, $b=-c a^{-2}$. 
  Note that these can happen simultaneously if and only if $c=-1$.
  
  We now assume $ab\neq 0$, $x\neq 0,-a$.
 Equation~\eqref{eq:thm6} is therefore
 \allowdisplaybreaks
 \begin{align}
 & (x+a)^{-2}-c x^{-2}=b,\text{ that is,}\nonumber\\
 & b (x+a)^2 x^2 -x^2+c(x+a)^2=0,\text{ that is,}\nonumber\\
 & b x^4-b a x^3+b a^2 x^2-x^2+c x^2-ac x+c a^2=0,\nonumber\\
 &x^4-a  x^3+\frac{ba^2+c-1}{b}x^2-\frac{ac}{b}x+\frac{ca^2}{b}=0.\label{eq:thm6-quartic}
  \end{align}
  
   Using $x=y+a$ in~\eqref{eq:thm6-quartic}, we obtain
  \begin{equation}
  \label{eq:main7}
  y^4+\frac{b a^2+c-1}{b}\, y^2+\frac{a(c+1)}{b}\, y
 +\frac{b a^4+(c-1) a^2}{b}=0.
  \end{equation} 
    Now, if $c=-1$ and $b=a^{-2}$, Equation~\eqref{eq:main7} is then
  \begin{equation}
  \label{eq:main7_1}
 y^4-a^2 y^2-a^4=0,
  \end{equation} 
    whose discriminant (of the underlying quadratic) is $2a^4$. Now, we know that the underlying quadratic of Equation~\eqref{eq:main7_1} has two distinct roots if and only if the above discriminant is a nonzero square. It is surely nonzero for $a\neq 0$, and we know that $2=-1$ is a perfect square, say $-1=\imath^2$ in $\F_{3^n}$ if and only if $n$ is even. For the quartic to have four distinct roots in $\F_{3^n}$, obtained from $y^2=  a^2(1\pm \imath)$ (these are the roots of the underlying quadratic), one needs $1\pm \imath$ to be a perfect square, which happens if $n\equiv 0\pmod 4$, and we provide the reason next.  
The minimal polynomial of one of the roots, say $\sqrt{1+\imath}$, is $x^4+x^2-1$ which has the roots $\{\alpha^2 + 2, 2\alpha^3 + 2\alpha^2 + 2\alpha,2\alpha^2 + 1,\alpha^3 + \alpha^2 + \alpha \}$  in $\F_{3^4}$ (we took the primitive polynomial $x^4-x^3-1$ with $\alpha$ as one of the roots), and consequently in all extensions $\F_{3^n}$ of $\F_{3^4}$ with $n\equiv 0\pmod 4$, and no others. Further, if $n\equiv 2\pmod 4$ we do not get more roots from Equation~\eqref{eq:main7_1} under $c=-1,b=a^{-2}$, therefore, in this case we only get the solutions $x=0,-a$.  
    
    If $c=-1$ and $b\neq a^{-2}$, we let $b=a$, obtaining the equation  $a^2 + a^5 + (1 + a^3) y^2 + a y^4=0$. The  discriminant of the underlying quadratic is $1+a^3=(1+a)^3$. Thus, if   $a=d^4-1$ for some random fixed $d\neq0$ such that $a\neq 0,1$, the previous equation has four solutions and so, the first claim is shown.
  
  If $c\neq \pm 1$, putting together the  potential solutions for Equation~\ref{eq:main7} and $y=0$ or $y=-a$, we see that we cannot have more than $5$ solutions. 
  
  We now show that the differential uniformity of $4$ is attained.  
 For $a=1,b=-1$,
 the quartic~\eqref{eq:main7} becomes
 $
y^4-(1+c) y^2-(1+c)y-(1+c)=0,
 $
 which is equivalent to $y^4-(c+1) (y^2+y+1)=0$, and further, $y^4-(c+1)(y-1)^2=0$ (observe that $y\neq 0,1$).
 Thus, taking $c=\alpha^2-1$, for some $\alpha\neq 0,1$, then $\frac{y^2}{y-1}=\pm \alpha$, which is equivalent to
 \[
 y^2\mp\alpha y\pm \alpha=0.
 \]
 These pairs of equations will each have two roots  if and only if the 
 $\eta(\alpha^2\pm 4\alpha)=\eta(\alpha^2\pm \alpha)=1$ ($\eta$ is the quadratic character), that is, if $\alpha^2\pm \alpha$ are nonzero squares in $\F_{3^n}$.
 However, we do have an argument for the existence of such $\alpha$, in general. We use~\cite[Theorem 5.48]{LN96}, which states that if $f(x)=a_2 x^2+a_1 x+a_0$ is a polynomial in a finite field $\F_q$ of odd characteristic, $a_2\neq 0, d=a_1^2-4a_0a_2$, and $\eta$ is the quadratic character on $\F_q$, then the Jacobsthal sum 
 \[
 \sum_{x\in\F_q} \eta(f(x))=\begin{cases}
 -\eta(a_2)&\text{ if }  d\neq 0\\
 (q-1)\eta(a_2)&\text{ if }  d= 0. 
 \end{cases}
 \]
 First, we take $f(x)=x^2-x$, $d=1$, and so, the Jacobsthal sum becomes 
 \begin{align*}
&   \sum_{x\in\F_{3^n}} \eta(x^2-x)=-\eta(1)=-1,\text {so,}\\
&-1=\eta(0)+\eta(1^2-1)+\sum_{x\in\F_{3^n},x\neq 0,1} \eta(x^2-x)=\sum_{x\in\F_{3^n},x\neq 0,1} \eta(x^2-x).
 \end{align*}
 Thus,
 \[
 \sum_{x\in\F_{3^n},x\neq 0,1} \eta(x^2-x)=-1.
 \]
 Similarly, 
  \[
 \sum_{x\in\F_{3^n},x\neq 0,-1} \eta(x^2+x)=-1.
 \]
   Let $N_1=|\{x|x\neq 0,\pm 1, \eta(x^2-x)=1\}$  and $N_2=|\{x|x\neq 0,\pm 1, \eta(x^2+x)=1\}$.
 We compute
 \begin{align*}
  -1&=\sum_{x\in\F_{3^n},x\neq 0,1} \eta(x^2-x)\\
  &=\eta((-1)^2-(-1))+\sum_{x\in\F_{3^n},x\neq 0,\pm 1} \eta(x^2-x)\\
  &=\eta(2)+\sum_{x\in\F_{3^n},x\neq 0,\pm 1} \eta(x^2-x)\\
  &=\eta(-1)+\sum_{x\in\F_{3^n},x\neq 0,\pm 1} \eta(x^2-x)=\eta(-1)+N_1-(3^n-3-N_1)\\
  -1&= \sum_{x\in\F_{3^n},x\neq 0,-1} \eta(x^2+x)\\
  &=\eta(1^2+1)+\sum_{x\in\F_{3^n},x\neq 0,\pm 1} \eta(x^2+x)\\
  &= \eta(-1)+\sum_{x\in\F_{3^n},x\neq 0,\pm 1} \eta(x^2+x)=\eta(-1)+N_2-(3^n-3-N_2).
 \end{align*}
 We therefore get
 \[
 N_1=N_2=\frac{3^n-4-\eta(-1)}{2}.
 \]
 The sets of $x\neq 0,\pm 1$ of cardinality $ \frac{3^n-4-\eta(-1)}{2}$ (equal to $\frac{3^n-3}{2}$ for $n$ odd and  $\frac{3^n-5}{2}$ for $n$ even) such that $x^2-x$, respectively, $x^2+x$ are squares, may still be disjoint. If they are not disjoint, we are done. Suppose now that the sets are disjoint; then,  we need to consider the Jacobsthal sum (again, using~\cite[Theorem 5.48]{LN96})
 \allowdisplaybreaks
 \begin{align*}
 & \sum_{x\in\F_{3^n},x\neq 0,\pm 1} \eta(x^2-x)\eta(x^2+x)
= \sum_{x\in\F_{3^n},x\neq 0,\pm 1} \eta((x^2-x)(x^2+x))\\
  &= \sum_{x\in\F_{3^n},x\neq 0,\pm 1}\eta(x^2(x^2-1))
   = \sum_{x\in\F_{3^n},x\neq 0,\pm 1} \eta(x^2)\eta(x^2-1)\\
  &= \sum_{x\in\F_{3^n},x\neq 0,\pm 1} \eta(x^2-1)
  =\sum_{x\in\F_{3^n}} \eta(x^2-1)-\eta(-1)=-1+\eta(-1),
 \end{align*}
 but that is impossible for $n\geq 3$, if $x^2-x$ and $x^2+x$ are never squares at the same time (given our prior counts, $N_1,N_2$).

 We now take $n$ to be an even integer.  We shall show that there are values of $c$ such that the $c$-differential uniformity is $5$. We take $b=a^{-2}$, so that $x=0$ is a solution of Equation (\ref{eq:thm6}). 
 The idea is to find, under this condition, $A,B$ such that the quartic in $y$ can be written as $(y^2+1)^2+A (y+B)^2=0$, where $B$ is a perfect square.
 We let $a$ such that $a^6+a^2+2=0$. This equation has the solutions $g+1,2g+2$  in $\F_{3^2}$, where $g$ is the primitive root vanishing $X^2-X-1=0$, and consequently, by field embedding, solvable in $\F_{3^n}$ for all even $n$.  Further, we take  $c=da^{-2}$, where $d=\frac{2 +2 a^6}{1 +2 a^2 + 2 a^4}$. Our quartic in $y$ becomes 
 \[
 y^4+ d y^2 + (a^3 + a d) y  +a^2 d=0,
 \]
 which we will write in the form
 \[
 (y^2 + 1)^2 + (1 + d) \left(y -\frac{ a^3+ad}{ 1+d}\right)^2=0.
 \]
 Observe that $\frac{ a^3+ad}{ 1+d}=-\frac{a^4+1}{a^3-a}$, when $d$ has the value we previously chose.
 Thus, assuming that $-(1+d)=-\frac{1 + a^4}{2 + a^2 + a^4}$ is a perfect square, say $\beta^2$ (we will check if that happens later), 
 the four solutions will be given by the equations (we let $\imath$ be the solution to $X^2=-1$ in $\F_{3^2}$, and any other even extension of $\F_{3^2}$) 
 \[\
 y^2+1= \pm\beta  \left(y +\frac{ a^4+1}{ a^3-a}\right),
 \]
 that is,
 \begin{align*}
 y^2\mp \beta y\mp \beta \frac{ a^4+1}{ a^3-a}+1=0.
 \end{align*}
 These last equations will have two solutions each if and only if 
 \[
 \displaystyle \beta^2\pm   \beta \frac{ a^4+1}{ a^3-a}-1
 \]
  is a perfect nonzero square. Using SageMath, we  quickly found values of $a$ satisfying the equation $a^6+a+2=0$ in $\F_{3^4}$  such that $-(1+d)=-\frac{1 + a^4}{2 + a^2 + a^4}$ and the expressions above are perfect squares, as well and, by field embedding, there are values of $a$ for every dimension divisible by $4$ where the last displayed expression is a square, as well.
  The theorem is shown.
   \end{proof}
 
   \begin{remark}
  Our computations in SageMath revealed that, there are other  values of the $c$-differential uniformity for the function in Theorem~\textup{\ref{thm:pn-3}}. In fact, if $c\neq \pm 1$ and $n=2$, then ${_c}\Delta_F=2;$  when $n=3$, we have ${_c}\Delta_F\in\{3,4\};$ for $n=4$, then ${_c}\Delta_F\in\{2,4,5\};$  for $n=5$, we  have ${_c}\Delta_F\in\{4\};$ if $n=6$, then ${_c}\Delta_F\in\{4,5\}$.  
  \end{remark}

  One might wonder if we can use the results of \textup{Bluher~\cite{Bluher04}} to investigate our quartic polynomial. Surely, one can removes the coefficient of $x^2$  in our quartic (to match Bluher's polynomial) using the substitution $b=\frac{1-c}{a^2}$, and the obtained polynomial is
  \[
 x^4-ax^3-\frac{a^3c}{1-c}x+\frac{a^4c}{1-c}
  \]
  but this, unfortunately, does not satisfy the needed condition $EA\neq B$ in the polynomial considered by  Bluher~\textup{\cite{Bluher04}}, $X^{q+1}+EX^q+AX+B$ ($q$ is a power of $p$; $p=3$ in our case), so we cannot use those methods. Even the more recent paper of Kim et al.~\textup{\cite{Kim19}} cannot be used,  because of the same reasons.

  We now look at item $(9)$ in Theorem~\ref{thm:HRS}, and prove a result for this function, and its generalization to $p$ odd (note that the case $p=3$ was proven in ~\cite{CS97,EFRST20}).
For $p$ prime and $n,k$ positive integers, we define the following parameter:
\begin{align*}
\ell& =\max\left\{\frac{1}{2}\gcd(p^k+1,p^n-1),\frac{1}{2}\gcd(p^k+1,p^n+1),\right.\\
&\qquad\qquad \left. \frac{1}{4}(\gcd(p^k+1,p^n-1)+\gcd(p^k+1,p^n+1))\right\}.
\end{align*}
  \begin{thm}
  \label{plusone}
  Let $p$ be an odd prime, and let $F(x)=x^{\frac{p^k+1}{2}}$ on $\F_{p^n}$, $1\leq k<n$, $n\geq 3$. If $c=-1$, then $F$ is PcN if and only if   $\frac{2n}{\gcd(2n,k)}$ is odd. 
  Otherwise, $F(x)$ will have the $(-1)$-differential uniformity  $_{-1}\Delta_F=\frac{p^{\gcd(k,n)}+1}{2}$.
  \end{thm}
  \begin{proof}
  We take the approach of~\cite{CS97,EFRST20}, where it was shown that $x^{\frac{3^k+1}{2}}$ is PcN on $\F_{3^n}$, for $c=\pm 1$. Similarly, our function is $PcN$ if and only if the $c$-derivative $(x+a)^{\frac{p^k+1}{2}}-c x^{\frac{p^k+1}{2}}$ is a permutation polynomial. With a change of variable $z=\frac{-a x}{4}$, we see that this is equivalent to  $\left(z-4\right)^{\frac{p^k+1}{2}}-c z^{\frac{p^k+1}{2}}$ being a permutation polynomial. Shifting by 2, we can see that this happens if and only if $h_c(z)=\left(z-2\right)^{\frac{p^k+1}{2}}-c \left(z+2\right)^{\frac{p^k+1}{2}}$ is a permutation polynomial.  
  We can always write $z=y+y^{-1}$, for some $y\in\F_{p^{2n}}$. Our condition (for general $c\neq 1$) becomes
  \allowdisplaybreaks
 \begin{align*}
 h_c(z)&=\left(y+y^{-1}-2\right)^{\frac{p^k+1}{2}}-c\left(y+y^{-1}+2\right)^{\frac{p^k+1}{2}}\\
 &= \frac{\left(y^2-2y+1\right)^{\frac{p^k+1}{2}}-c\left(y^2+2y+1\right)^{\frac{p^k+1}{2}}}{y^{\frac{p^k+1}{2}}}\\
&= \frac{\left(y-1\right)^{p^k+1}-c\left(y+1\right)^{p^k+1}}{y^{\frac{p^k+1}{2}}}\\
 &= \frac{(1-c)y^{p^k+1}-(1+c)y^{p^k}-(1+c) y+(1-c)}{y^{\frac{p^k+1}{2}}}\\
 &=(1-c) y^{\frac{p^k+1}{2}}-(1+c) y^{\frac{p^k-1}{2}}-(1+c) y^{\frac{-p^k+1}{2}}+(1-c) y^{\frac{-p^k-1}{2}}\\
 &= (1-c) D_{\frac{p^k+1}{2}}\left(z\right)-(1+c) D_{\frac{p^k-1}{2}}\left(z\right)
 \end{align*}
 is a permutation polynomial, where $D_m(x)(=D_m(x,1))$ is the Dickson polynomial of the first kind, in our notation. 
 
  If $c=-1$, we obtain that
 $D_{\frac{p^k+1}{2}}(x)$ must be a permutation polynomial, and this is equivalent (by~\cite{No68}) to   $\displaystyle \gcd\left( \frac{p^k+1}{2}, p^{2n}-1 \right)=1$. This last identity can be further simplified to $\displaystyle \gcd\left( p^k+1 , p^{2n}-1 \right)=2$. By Lemma~\ref{lem:gcd} a necessary and sufficient condition for that to happen is for   $\displaystyle \frac{2n}{\gcd(2n,k)}$ to be odd (this holds if and only if $k$ is even and, if $n=2^ta,\,2\!\not| a$, with $t\geq0$, and $k=2^\ell b,\,2\!\not| b$, then $\ell\geq t+1$).
 
In general, if $\ell=\max\left\{\left|D^{-1}_{\frac{p^k+1}{2}}(b)\right|: b\in \F_{p^n}\right\}$, then $_{-1}\Delta_F=\ell$. Here we will use Theorem~9 of~\cite{CGM88} stated above. Here $d=\frac{p^k+1}{2}$, so $m=\gcd\left(\frac{p^k+1}{2},p^n-1\right)=\frac{1}{2}\gcd(p^k+1,p^n-1),\,\bar{\ell}=\gcd\left(\frac{p^k+1}{2},p^n+1\right)=\frac{1}{2}\gcd(p^k+1,p^n+1)$. Then, in general, we can obtain all cases, 
and so 
\allowdisplaybreaks
\begin{align*}
\ell& =\max\left\{\frac{1}{2}\gcd(p^k+1,p^n-1),\frac{1}{2}\gcd(p^k+1,p^n+1),\right.\\
&\qquad\qquad \left. \frac{1}{4}(\gcd(p^k+1,p^n-1)+\gcd(p^k+1,p^n+1))\right\}\\
&=\max\left\{\frac{1}{2}\gcd(p^k+1,p^n-1),\frac{1}{2}\gcd(p^k+1,p^n+1)\right\}.
\end{align*}
Note that, by Lemma \ref{lem:gcd}, $\gcd(p^k+1,p^n-1)=2$ if $\frac{n}{\gcd(n,k)}$  is odd and $\gcd(p^{k}+1,p^n-1)=p^{\gcd(k,n)}+1$ if $\frac{n}{\gcd(n,k)}$ is even, while 
$$\begin{array}{ll}&\gcd(p^k+1,p^n+1)=\gcd(p^k+1,p^n+1-(p^k+1))\\
&=\gcd(p^k+1,p^{k}(p^{n-k}-1))=\gcd(p^k+1,p^{n-k}-1)\\
&=\left\{\begin{array}{l}
2 \mbox{ if }\frac{n-k}{\gcd(n-k,k)}=\frac{n-k}{\gcd(n,k)} \mbox{  is odd}\\
p^{\gcd(k,n)}+1\mbox{ if }\frac{n-k}{\gcd(n-k,k)}=\frac{n-k}{\gcd(n,k)} \mbox{ is even}
\end{array}\right.
 \end{array}$$ so, if $n=2^ta,\,2\!\not| a$, with $t\geq0$, and $k=2^\ell b,\,2\!\not| b$, then, if $\ell\geq t+1$, then $m=1=\bar{\ell}$, and so $\ell=1$ and the function is PcN, while, if $\ell=t$, then $m=2,\,\bar{\ell}=\frac{p^{\gcd(k,n)}+1}{2}=\ell$, while, if $\ell< t$, then $m=\frac{p^{\gcd(k,n)}+1}{2},\,\bar{\ell}=2$, and so $\ell=\frac{p^{\gcd(k,n)}+1}{2}$. Summarizing, if $\ell\leq t$, then  $\ell=\frac{p^{\gcd(k,n)}+1}{2}$.

The proof of the theorem is complete.
    \end{proof}

\begin{thm}    
 Let $G(x)=x^{\frac{3^n+3}{2}}$ on $\F_{3^n}$, $n\geq 2$. If $c=-1$, then $G$ is PcN if $n$ is odd, APcN if $n$ is even. If $c=1$, then its differential uniformity is ${_1}\Delta_G=1$, if $n$ is even and ${_1}\Delta_G=4$ if $n$ is odd.
\end{thm}
\begin{proof} 
Now, $G$ is PcN on $\F_{3^n}$, for  $c\neq 1$, if and only if (by similar arguments as before) the following is a permutation polynomial (where $z=y+y^{-1}, y\in \F_{3^{2n}}$), 
\allowdisplaybreaks
 \begin{align*}
 h_c(z)&=\left(y+y^{-1}-2\right)^{\frac{3^n+3}{2}}-c\left(y+y^{-1}+2\right)^{\frac{3^n+3}{2}}\\
 &= \frac{\left(y-1\right)^{3^n+3}-c\left(y+1\right)^{3^n+3}}{y^{\frac{3^n+3}{2}}}\\ 
 &= \frac{y^{3^n+3}-y^{3^n}-y^3+1-c\left(y^{3^n+3}+y^{3^n}+y^3+1\right)}{y^{\frac{3^n+3}{2}}}\\  
&= \frac{(1-c) \left(y^{3^n+3}+1\right)-(1+c)(y^{3^n}+y^3)}{y^{\frac{3^n+3}{2}}}\\
&=(1-c) D_{\frac{3^n+3}{2}}(z)-(1+c) D_{\frac{3^n-3}{2}}(z).
 \end{align*}

 If $c=-1$, then $h_c(z)= 2 \,D_{\frac{3^n+3}{2}}(z)$, which is known to be a permutation polynomial if and only if $\gcd\left(\frac{3^n+3}{2}, 3^{2n}-1 \right)=1$, which is easy to show that it always happens for $n$ odd (for $n$ even,  $\gcd\left(\frac{3^n+3}{2}, 3^{2n}-1 \right)=2$). This can be seen from the following argument:
\begin{align*}
\gcd\left(\frac{3^n+3}{2}, 3^{2n}-1 \right)&=\frac{1}{2}\gcd(3^n+3,3^{2n}-1)\\
& =\frac{1}{2}\gcd(3^n+3,3^{2n}-1-(3^n+3)(3^n-1))\\
& =\frac{1}{2}\gcd(3^n+3,2-2\cdot3^n)\\
&=\frac{1}{2}\gcd(3(3^{n-1}+1),2(1-3^n))\\
& =\frac{1}{2}\gcd(3^{n-1}+1,3^n-1).
\end{align*}
 Now, by Lemma \ref{lem:gcd},
 $\gcd(3^{n-1}+1,3^n-1)=2$ if $\frac{n}{\gcd(n-1,n)}=n$ is odd, and $\gcd(3^{n-1}+1,3^n-1)=3^{\gcd(n-1,n)}+1=4$ if $\frac{n}{\gcd(n-1,n)}=n$ is even, which implies our claim.
  
Thus, for $n$ odd and $c=-1$, the function is PcN. Let $n$ be even:
 if $\ell=\max\left\{\left|D^{-1}_{\frac{{3^n+3}}{2}}(b)\right|: b\in \F_{p^n}\right\}$, then $_{-1}\Delta_F=\ell$. By Theorem~9 of~\cite{CGM88}, supposing $2^r||(3^{2n}-1)$, then, for $x_0\in\F_{3^n}$,
 \begin{equation*}
\label{eq:Dickson}
|D_d^{-1}(D_d(x_0))|=
\begin{cases}
m, &\mbox{ if }\eta(x_0^2-4)=1,\,D_d(x_0)\neq\pm2\\
\bar{\ell}, &\mbox{ if }\eta(x_0^2-4)=-1,\,D_d(x_0)\neq\pm2\\
\frac{m}{2}, &\mbox{ if }\eta(x_0^2-4)=1,\,2^t||d,\,1\leq t\leq r-2,D_d(x_0)=-2\\
\frac{\bar{\ell}}{2}, &\mbox{ if }\eta(x_0^2-4)=-1,\,2^t||d,\,1\leq t\leq r-2,D_d(x_0)=-2\\
\frac{m+\bar{\ell}}{2}, &\mbox{ otherwise,}
\end{cases}
\end{equation*}
where $m=\gcd(\frac{3^n+3}{2},3^n-1)=\frac{1}{2}\gcd(3^n+3,3^n-1)=\frac{1}{2}\gcd(3(3^{n-1}+1),3^n-1)=2,\,\bar{\ell}=\gcd(\frac{3^n+3}{2},3^n+1)=2$, all under $n$ being even. Thus, $\ell=2$.

When $c=1$, we can apply~\cite[Theorem 3]{HS97}, since $\frac{3^n+3}{2}=\frac{3^n-1}{2}+2$ and so, ${_1}\Delta_G=1$, if $n$ is even and ${_1}\Delta_G=4$ if $n$ is odd.
We can also show the result easily by a similar argument as above. We need to look at $D_{\frac{3^n-3}{2}}(z)$ and its value sets.
First, observe that (it is easy to show by algebraic manipulations, as in our previous discussion) that 
\begin{equation}
\label{eq:gcd2}
\gcd\left(\frac{3^n-3}{2}, 3^{2n}-1 \right)
=\begin{cases}
1 &\text{ if } n\equiv 0\pmod 2\\
4 &\text{ if } n\equiv 3\pmod 4\\
8 &\text{ if } n\equiv 1\pmod 4.
\end{cases}
\end{equation}

 As in the case $c=-1$, we shall be using Theorem~9 of~\cite{CGM88}. In this case,
  $m=\gcd(\frac{3^n-3}{2}, 3^n-1)=\gcd(\frac{3(3^{n-1}-1)}{2}, 3^n-1)=\gcd(\frac{3^{n-1}-1}{2},3^n-1)=1,2$, for $n$ even, respectively, odd; $\bar{\ell}=\gcd(\frac{3^n-3}{2}, 3^n+1)=1,4,8$, as in Equation~\eqref{eq:gcd2}. Thus, if $n\equiv 0\pmod 2$, we get the maximum cardinality of the preimage~\eqref{eq:Dickson} of our Dickson polynomial to be 1; and if $n$ is odd, the maximum cardinality of the preimage~\eqref{eq:Dickson} of our Dickson polynomial is  $\frac{8}{2}=4$, matching therefore the result from~\cite[Theorem 3]{HS97}.
The theorem is shown.
\end{proof}

\section{Concluding remarks}

In this paper we investigated the $c$-differential uniformity of the Gold function and some of the functions from the Helleseth-Rong-Sandberg table and related ones. 
For example, we showed that $x^{\frac{p^k+1}{2}}$ is PcN with respect to $c=-1$ if and only if $\frac{2n}{\gcd(2n,k)}$ is odd, as well $x^{\frac{3^n+3}{2}}$ is PcN, respectively, APcN, with respect to $c=-1$, if $n$ is odd, respectively, even.
Surely, it would be interesting to continue with some of the other entries in the Helleseth-Rong-Sandberg table, or the results from Dobbertin et al.~\cite{Dob03}, or even to find newer PN or APN classes of functions, through the prism of the newly defined $c$-differentials concept we introduced in~\cite{EFRST20}. 

\vskip.3cm
\noindent
{\bf Acknowledgments}. The authors are grateful to the reviewers for  pointing out some initial errors in the proof of Theorem~\ref{thm:pn-3}, and for the very helpful comments and suggestions which have highly improved the manuscript.


\begin{thebibliography}{99}

\bibitem{BT} D. Bartoli, M. Timpanella, {\it On a generalization of planar functions}, J. Algebr. Comb. (2019).
 
\bibitem{Bluher04}
A.W. Bluher, {\em On $x^{q+1} + ax + b$}, Finite Fields Appl. 10 (3) (2004), 285--305.

\bibitem{BCJW02} N. Borisov, M. Chew, R. Johnson, D. Wagner, {\em Multiplicative Differentials}, In: Daemen J., Rijmen V. (eds) Fast Software Encryption. FSE 2002. LNCS 2365. Springer, Berlin, Heidelberg, 2002.

\bibitem{Bud14}
L. Budaghyan, {\em Construction and Analysis of Cryptographic Functions}, Springer-Verlag, 2014.

  \bibitem{CH1} C.~Carlet, {\em Boolean functions for cryptography and error correcting codes}, In: Y. Crama, P. Hammer  (eds.), Boolean Methods and Models,
Cambridge Univ. Press, Cambridge, pp. 257--397, 2010.

\bibitem{CH2}
C. Carlet, {\em Vectorial Boolean Functions for Cryptography},
In: Y. Crama, P. Hammer  (eds.), Boolean Methods and Models,
Cambridge Univ. Press, Cambridge, pp. 398--472, 2010.

\bibitem{Car18}
C. Carlet, {\em Characterizations of the Differential Uniformity of
Vectorial Functions by the Walsh Transform}, IEEE Trans. Inf. Theory 64:9 (2018), 6443--6453.
 

\bibitem{CV95}
F. Chabaud, S. Vaudenay, {\em Links between differential and linear cryptanalysis}, In: Adv. in Crypt -- EUROCRYPT' 94, LNCS 950, pp. 356--365, 1995.

\bibitem{CGM88} W. Chou, J. Gomez-Calderon, G. L.Mullen, {\em Value sets of Dickson polynomials over finite fields}, J. Number Theory 30:3 (1988),  334--344.
 
\bibitem{CS97}
R. S. Coulter, R. W. Matthews,
{\em Planar functions and planes of Lenz-Barlotti class} II, Des. Codes Cryptogr. 10 (1997), 167--184.

\bibitem{CS17} T. W.~Cusick, P.~St\u anic\u a,
{Cryptographic Boolean Functions and Applications} (Ed. 2), Academic Press, San Diego, CA,  2017.
 

\bibitem{Dob06}
H. Dobbertin, P. Felke, T. Helleseth, P. Rosendahl, {\em Niho type cross-correlation functions via Dickson polynomials and Kloosterman sums}, IEEE Trans. Inf. Theory 52 (2) (2006) 613--627.

\bibitem{Dob03}
H. Dobbertin, D. Mills, E. N. Muller, A. Pott, and W. Willems, {\em APN
functions in odd characteristic}, Discr. Math. 267 (1-3) (2003), 95--112.

\bibitem{EFRST20}
P. Ellingsen, P. Felke, C. Riera P. St\u anic\u a, A. Tkachenko,
{\em $C$-differentials, multiplicative uniformity and (almost) perfect $c$-nonlinearity}, to appear in IEEE Trans. Inf. Theory, 2020.

\bibitem{HK08}
T. Helleseth, A. Kholosha, {\em On the equation $x^{2^\ell+1}+x+a=0$ over $GF(2^k)$}, Finite Fields Appl. 14 (2008), 159--176.

\bibitem{HRS99}
T. Helleseth, C. Rong,  D. Sandberg, {\em New families of almost perfect
nonlinear power mappings}, IEEE Trans. Inf. Theory 45 (1999), 475--485.

\bibitem{HS97}
T. Helleseth, D. Sandberg, {\em Some power mappings with low
differential uniformity}, Appl. Algebra Eng. Commun. Comput. 8 (1997),
363--370.

\bibitem{Kim19}
K. H. Kim, J. Chloe, S. Mesnager, {\em Solving $X^{q+1}+X+a$ over Finite Fields}, \url{https://eprint.iacr.org/2019/1493.pdf}.

\bibitem{Leducq12} E. Leducq, {\em New families of APN functions in characteristic 3 or 5}, Contemporary Mathematics
Volume 574, 2012.

\bibitem{LN96} R. Lidl, H. Niederreiter, { Finite fields (Encyclopedia of Mathematics and Its Applications)}. Cambridge: Cambridge University Press, 1996.


\bibitem{MesnagerBook} S. Mesnager, { Bent functions: fundamentals and results}, Springer Verlag, 2016.

\bibitem{No68}
W. N\"obauer, {\em \"Uber eine Klasse von Permutationspolynomen und die dadurch dargestellten Gruppen}, J. Reine Angew. Math. 231 (1968), 215--219.
 
\bibitem{Tok15} N. Tokareva, { Bent Functions, Results and Applications to Cryptography}, Academic Press, San Diego, CA,  2015.

\bibitem{YZ20}
H. Yan, Z. Zhou,
{\em Power Functions over Finite Fields with Low $c$-Differential Uniformity}, https://arxiv.org/pdf/2003.13019.pdf.

\bibitem{ZW11} Z. Zha, X. Wang, 
{\em Almost Perfect Nonlinear Power Functions in Odd Characteristic}, IEEE Trans. Inf. Theory 57:7 (2011) (1999), 4826--4832.
 
 \end{thebibliography}
 \end{document}